\documentclass[letterpaper, 10 pt, conference]{ieeeconf}  
\usepackage{gen_settings}

\IEEEoverridecommandlockouts                              

\title{\LARGE \bf
Lipschitz Continuity of Signal Temporal Logic Robustness Measures: \\
Synthesizing Control Barrier Functions from One Expert Demonstration
}

\author{Prithvi Akella$^{*}$, Apurva Badithela$^{*}$, Richard Murray, and Aaron D. Ames$^{1}$
\thanks{$^*$ The authors contributed equally. Corresponding authors: P. Akella, A. Badithela, \texttt{\{pakella,apurva\}@caltech.edu}. }
\thanks{We acknowledge funding from AFOSR Test and Evaluation Program, grant FA9550-19-1-0302.}
\thanks{$^{1}$All authors are with the California Institute of Technology}
}

\begin{document}

\maketitle
\thispagestyle{empty}
\pagestyle{empty}

\begin{abstract}
Control Barrier Functions (CBFs) allow for efficient synthesis of controllers to maintain desired invariant properties of safety-critical systems. However, the problem of identifying a CBF remains an open question. As such, this paper provides a constructive method for control barrier function synthesis around one expert demonstration that realizes a desired system specification formalized in Signal Temporal Logic (STL). First, we prove that all STL specifications have Lipschitz-continuous robustness measures. Second, we leverage this Lipschitz continuity to synthesize a time-varying control barrier function. By filtering control inputs to maintain the positivity of this function, we ensure that the system trajectory satisfies the desired STL specification. Finally, we demonstrate the effectiveness of our approach on the Robotarium.
\end{abstract}

\section{INTRODUCTION}
Control barrier functions (CBFs) have emerged as a pre-eminent tool for controller synthesis for safety-critical systems~\cite{ames2016control,hsu2015control,xiong2021slip,srinivasan2018control,zeng2021safety,choi2020reinforcement}, due to their facilitation of robust, input-to-state-safe control~\cite{kolathaya2018input,alan2021safe,tezuka2020time}, multi-agent control~\cite{cai2021safe,tan2021distributed}, and learning-based control~\cite{marvi2021safe,ma2021model}, among other applications. In particular, CBFs are used to synthesize controllers that maintain forward invariance of a subset \(\mathcal{C}\) of the system state space. Despite their tremendous success, synthesizing CBFs remains an open question with tremendous interest~\cite{jagtap2020control,robey2020learning,srinivasan2020synthesis,lindemann2018control,lindemann2020barrier,huang2020multi}. Synthesis of CBFs is difficult, in large part because of two requirements these functions must satisfy: i)  unit relative degree with respect to the controller, which is addressed in part via exponential barrier functions~\cite{nguyen2016exponential,azimi2021exponential}, and ii) the differential inequality contingent on system dynamics that must hold over the entire state space~\cite{ames2016control}. 

These conditions can be expressed as constraints on optimization problems over function spaces, leading to learning-based approaches to identify these functions against provided models~\cite{jagtap2020control,robey2020learning,srinivasan2020synthesis,dawson2022safe}. However, these approaches tend to require a large amount of system data and suffer from a lack of generality to adapt to circumstances not reflected in the underlying data. To address these issues, recent work aims to construct these functions from first principles based on the expression of system behaviors in signal temporal logic (STL)~\cite{lindemann2018control,lindemann2020barrier,huang2020multi,donze2010robust}. As such, they provide constructive approaches to transition between desired system behaviors and a quantifiable encoding of those behaviors through barrier functions.  However, the synthesis of barrier functions developed in these prior works is not easily adaptable to changing scenarios --- an issue we aim to resolve.
\begin{figure}[t]
    \centering
    \includegraphics[width = \columnwidth]{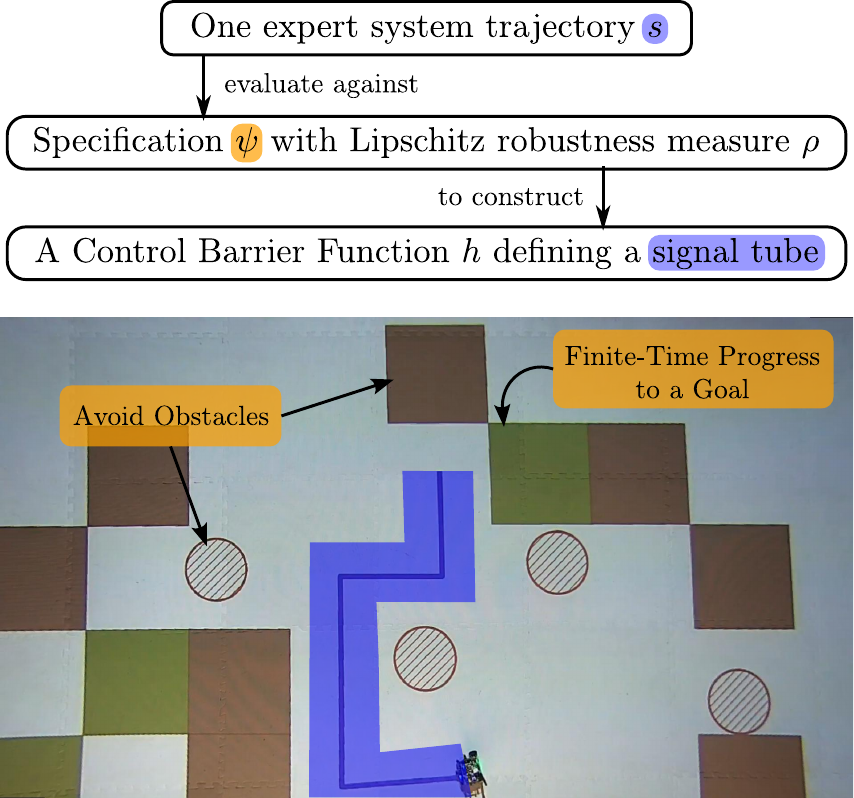} \vspace{-0.2 in}
    \caption{An overview of our proposed approach.  By evaluating one expert trajectory against the Lipschitz-continuous robustness measure $\rho_{\psi}$ for an STL specification $\psi$, we construct a time-varying CBF defining a signal tube of satisfying trajectories.  By filtering control inputs using this barrier function, we guarantee specification satisfaction.  The expert signal and tube are shown in blue and the specification is shown in gold highlights.} \vspace{-0.2 in}
    \label{fig:overview}
\end{figure}

\newidea{Our Contribution:} Our work extends prior works~\cite{lindemann2018control,lindemann2020barrier,huang2020multi} by constructively synthesizing control barrier functions for a larger fragment of signal temporal logic specifications. Furthermore, we bridge the gap between the first principles approach of~\cite{lindemann2018control,lindemann2020barrier} and learning-based approaches by constructing our CBF using one expert demonstration. Thus, our method is amenable to rapid synthesis of CBFs online and can adapt to changing environments while maintaining desired safe system behavior throughout. We demonstrate our synthesis procedure extensively in simulation and present a hardware demonstration on the Georgia Tech Robotarium~\cite{wilson2020robotarium}. The synthesized CBFs succeed in safely steering our system despite randomized test cases produced via the methods described in~\cite{akella2022sample,akella2022scenario}. These successes indicate the utility of our approach insofar as we are rapidly able to synthesize control barrier functions online for complex signal temporal logic specifications and realize the desired behavior reliably on hardware.

\newidea{Structure:} Section~\ref{sec:barriers_and_stl} reviews control barrier functions and signal temporal logic and states the problem under study --- synthesis of CBFs to satisfy an STL formula.  Section~\ref{sec:lipschitz} details our efforts to prove that all STL specifications have Lipschitz continuous robustness measures. Section~\ref{sec:cbf_synt} leverages this Lipschitz continuity to construct time-varying CBFs provided the existence of one expert demonstration satisfying the specification of interest.  Finally, Section~\ref{sec:examples} details the implementation of our approach in both simulation and hardware on the Georgia Tech Robotarium~\cite{wilson2020robotarium}.

\section{Preliminaries}
\label{sec:barriers_and_stl}
Before formally stating the problem, we provide a brief background on control barrier functions and signal temporal logic as they are central topics for the remainder of this paper.

\spacing
\newidea{Control Barrier Functions:} Inspired by barrier methods in optimization (see Chapter 3 of \cite{forsgren2002interior}), control barrier functions are a tool used to ensure safety in safety-critical systems that are control-affine, \textit{i.e.} with state $x$ and control input $u$:
\begin{equation}
    \label{eq:base_system}
    \dot x = f(x) + g(x) u, \quad x \in \mathcal{X} \subseteq \mathbb{R}^n,~u\in\mathcal{U}\subseteq \mathbb{R}^m,
\end{equation}
Here, $\mathcal{X}$ is the state space, and $\mathcal{U}$ is the input space.  
Time-varying control barrier functions (TVCBF) are defined against an inequality using class-$\mathcal{K}_e$ functions $\alpha: \mathbb{R} \to \mathbb{R}$.  These functions are such that $\alpha(0) = 0$ and $r \geq s \iff \alpha(r) \geq \alpha(s)$. Then, a time-varying control barrier function is defined as follows.

\begin{definition}
    \label{def:time_varying_cbf}
    A \textit{time-varying control barrier function} \({h:\mathbb{R}^n \times \mathbb{R}_{\geq 0} \to \mathbb{R}}\) is a differentiable function satisfying the following inequality $\forall~(x,t) \in \mathcal{X} \times \mathbb{R}_{\geq 0}$ and for some class-$\mathcal{K}_e$ function $\alpha$:
    \begin{equation}
        \label{eq:tv_cbf_criteria}
        \sup_{u \in \mathbb{R}^m}~\left[\frac{\partial h}{\partial x}\left( f(x) + g(x)u\right) + \frac{\partial h}{\partial t}\right] \geq -\alpha(h(x,t)).
\end{equation}
\end{definition}
\noindent To better explain the utility of these functions, consider the $0$-superlevel set of a time-varying control barrier function $h$:
\begin{equation}
    \label{eq:zerosuplevel}
\mathcal{C} = \left\{(x,t) \in \mathcal{X} \times \mathbb{R}_{\geq 0}~|~h(x,t) \geq 0\right\},
\end{equation}
and the set of control inputs satisfying the inequality~\eqref{eq:tv_cbf_criteria}:
\begin{equation}
    \label{eq:valid_input_set}
    K_{cbf}(x,t) = \left\{u \in \mathcal{U}~\big|~\dot h(x,u,t) \geq -\alpha(h(x,t))\right\}.
\end{equation}
Finally, let $\phi^U_t(x_0)$ denote the flow of the control affine system~\eqref{eq:base_system} from the initial condition $x_0 \in \mathcal{X}$ when steered by a controller $U:\mathcal{X} \times \mathbb{R}_{\geq 0} \to \mathcal{U}$, \textit{i.e.} $\phi^U_0(x_0) = x_0$ and,
\begin{equation}
    \label{eq:solution}
    \dot \phi^U_t(x_0) = f\left(\phi^U_t(x_0) \right) + g\left(\phi^U_t(x_0)\right)U\left(\phi^U_t(x_0),t\right).
\end{equation}
Then the utility of these functions is expressed formally in the following theorem regarding forward invariance of the $0$-superlevel set $\mathcal{C}$:
\begin{theorem}\label{thm:cbf_condition}
    Let $I(x_0) = [0,t_{\max}]$ be the interval of existence of solutions $\phi^U_t(x_0)$, let $K_{cbf}(x,t)$ be as defined in~\eqref{eq:valid_input_set}, and let $\mathcal{C}$ be as per~\eqref{eq:zerosuplevel}.  If $\forall \, t \in I(x_0),~U\left(\phi^U_t(x_0),t\right) \in K_{cbf}\left(\phi^U_t(x_0),t\right)$, then $\forall \,t \in I(x_0),~(\phi^U_t(x_0),t) \in \mathcal{C}$.
\end{theorem}
\spacing
\newidea{Signal Temporal Logic:} 
Signal Temporal Logic (STL) is a logic that allows for specifying trace properties of dense-time, real-valued signals~\cite{donze2010robust}. STL formulas are defined recursively as follows,
\begin{equation}
    \label{eq:spec}
    \psi ::= \mu~|~\neg \psi~|~\psi_1 \lor \psi_2 ~|~ \psi_1 \wedge \psi_2 ~|~
    \psi_1 \until_{[a,b]} \psi_2,
\end{equation}
where $a,b\in\mathbb{R}_{\geq 0} \cup \{\infty\}$, $a\leq b$, and $\mu$ denotes an atomic proposition which is evaluated over states $x\in\mathbb{R}^n$ and returns a Boolean value if a function $h_{\mu}: \mathbb{R}^n \to \mathbb{R}$ is positive at $x$:
\begin{equation}
    \label{eq:propos_def}
    \mu(x) = \true \iff h_{\mu}(x)\geq 0.
\end{equation}

\noindent We denote signals $s: \mathbb{R}_{\geq 0} \to \mathbb{R}^n$ and the space of all signals $\signalspace = \{s~|~s:\mathbb{R}_{\geq 0} \to \mathbb{R}^n\}$.  Then, we denote that a signal $s$ satisfies $\psi$ at time $t$ via the notation $(s,t) \models \psi$.  The satisfaction operator $\models$ is defined recursively as follows:
\begin{equation}
    \begin{aligned}
        & (s,t) \models \mu & \hspace{-0.1 in} \iff & \mu(s(t)) = \true, \\
        & (s,t) \models \neg \psi & \hspace{-0.1 in}  \iff & (s,t) \not \models \psi, \\
        & (s,t) \models \psi_1 \lor \psi_2 & \hspace{-0.1 in}  \iff & (s,t) \models \psi_1 \lor (s,t) \models \psi_2, \\
        & (s,t) \models \psi_1 \wedge \psi_2 & \hspace{-0.1 in}  \iff & (s,t) \models \psi_1 \wedge (s,t) \models \psi_2, \\
        & (s,t) \models \psi_1 \until_{[a,b]} \psi_2 & \hspace{-0.1 in}  \iff & \exists~t' \in [t+a,t+b] \suchthat \\
        & & & \left((s,t') \models \psi_2\right)~\wedge \\
        & & & \left(\forall~t''\in[t+a,t']~(s,t'') \models \psi_1\right).
    \end{aligned}
\end{equation}
Furthermore, every STL specification $\psi$ comes equipped with a robustness measure $\rho$ that evaluates signals $s \in \signalspace$.  If for some time $t \in \mathbb{R}_{\geq 0}$, $\rho(s,t) \geq 0$, then the signal $s$ satisfies the specification $\psi$~\cite{baier2008principles,donze2010robust,fainekos2009robustness, maler2004monitoring}.
\begin{definition}
\label{def:robustness}
A function $\rho_{\psi}: \signalspace \times \mathbb{R}_+ \to \mathbb{R}$ is a \textit{robustness measure} for an STL specification $\psi$ if it satisfies the following equivalency: $\rho_{\psi}(s,t) \geq 0 \iff (s,t) \models \psi$.
\end{definition}
\begin{example}
Let $\psi = \neg(\true \until_{[0,2]} (|s(t)| > 2))$, then any real-valued signal $s:\mathbb{R}_{\geq 0} \to \mathbb{R}$ satisfies $\psi$ at time $t$, \textit{i.e.} $(s,t) \models \psi$ if $\forall~t' \in [t,t+2],~|s(t')| \leq 2$.  The corresponding robustness measure $\rho_{\psi}(s,t) = \min_{t' \in [t,t+2]}~2-|s(t')|$.
\end{example}
\noindent While robustness measures defined according to Definition~\ref{def:robustness} align with the proposition definition in equation~\eqref{eq:propos_def} and prior works~\cite{lindemann2018control,lindemann2020barrier}, it is not the only way of defining such a measure, \textit{e.g.} see Definition 3 of~\cite{donze2010robust} or Section 2.3 in~\cite{fainekos2009robustness}. 

\newidea{Formal Problem Statement:} Our goal is to synthesize a time-varying control barrier function as per Definition~\ref{def:time_varying_cbf} whose continued positivity ensures system satisfaction of an STL specification as per equation~\eqref{eq:spec}. Denoting closed loop trajectories as per~\eqref{eq:solution}, our formal problem statement follows.
\begin{problem}
    \label{prob}
    Let $\psi$ be an STL specification as per~\eqref{eq:spec}.  For a nonlinear control-affine system~\eqref{eq:base_system}, construct a time-varying control barrier function $h$ according to Definition~\ref{def:time_varying_cbf} such that for some time $T \in \mathbb{R}_{\geq 0}$ and signal $z \in \signalspace$:
    \begin{equation}
        \label{eq:desired_condition}
        h(z(t), t) \geq 0,\,~\forall~t \in [0,T] \implies z \models \psi.
    \end{equation}
\end{problem}

\spacing
\newidea{Key Idea:} To address this problem, we first prove in Section~\ref{sec:lipschitz}, that every STL specification of the form in~\eqref{eq:spec} has a Lipschitz-continuous robustness measure $\rho$. That is, for some fixed time $t$ and any two signals $s,z \in \signalspace$, there exists a constant $L \geq 0$ such that $|\rho(s,t) - \rho(z,t)| \leq L \|s-z\|$.  Here, the (semi)norm $\|\cdot\|$ is defined on the $\signalspace$ and will be formally defined in Definition~\ref{def:signal_norm}.
Second, we collect one expert signal $s$ that satisfies our specification $\psi$ and define a time-varying CBF $h$ whose $0$-superlevel corresponds to a signal tube around the expert trajectory $s$. Then, the invariance of this $0$-superlevel set corresponds to the state trajectory remaining within this tube of signals that satisfy the desired specification. This step is detailed in Section~\ref{sec:cbf_synt}.

\section{Lipschitz Robustness Measures}
\label{sec:lipschitz}
The lemmas and theorems in this section will follow an inductive argument to prove that all specifications $\psi$~\eqref{eq:spec} have Lipschitz-continuous robustness measures.  Our first lemma will prove the base case --- that all STL atomic propositions $\mu$ as per~\eqref{eq:propos_def} have Lipschitz robustness measures $h_{\mu}$.
\begin{lemma}
\label{lem:predicate}
Let $\mu$ be an STL proposition as defined in~\eqref{eq:propos_def}.  There exists a Lipschitz continuous function $h_{\mu}:\mathbb{R}^n \to \mathbb{R}$ that satisfies the following two conditions:
\begin{itemize}
    \item $\mu(x) = \true \iff h_{\mu}(x) \geq 0,$
    \item $|h_{\mu}(x) - h_{\mu}(z)| \leq L \|x-z\|,~x,z\in\mathcal{X}$, for some $L \geq 0$.
\end{itemize}
\end{lemma}
\begin{proof}
First, for each proposition $\mu$, we define $\llbracket \mu \rrbracket$ as the set where $\mu$ holds and $\ell_{\mu}$ as the $2$-norm distance to that set.
\begin{equation}
    \llbracket \mu \rrbracket = \{x \in \mathcal{X}~|~\mu(x) = \true \},~\ell_{\mu}(x) = \min_{y \in \llbracket \mu \rrbracket}~\|x-y\|.
\end{equation}
Then, $h_{\mu}(x) = -\ell_{\mu}(x)$ satisfies both required conditions.
\end{proof}

\noindent The next two lemmas act as inductive arguments beyond the base case.  The first will prove that if STL specifications $\psi_1,\psi_2$ have Lipschitz robustness measures, then logical combinations of these specifications have a Lipschitz robustness measure as well.  To facilitate the statement of this and future lemmas, a (semi)-norm over $\signalspace$ is defined as follows:
\begin{definition}
    \label{def:signal_norm}
    Let $s$ be a signal in the signal space $\signalspace$.  We define a \textit{signal (semi)norm} $\|\cdot\|_{[a,b]}$ for times $a,b \in \mathbb{R}_{\geq 0} \cup \{\infty\}$ as follows, with $\|\cdot\|$ the $2$-norm on $\mathbb{R}^n$:
    \begin{equation}
        \|s\|_{[a,b]} = \min_{t \in [a,b]}\|s\|.
    \end{equation}
\end{definition}
\begin{lemma}
    \label{lem:logical_conjunctions}
    Let $\psi_1,\psi_2$ be STL specifications as per~\eqref{eq:spec} such that their corresponding robustness measures $\rho_{\psi_1},\rho_{\psi_2}$ are Lipschitz with constants $L_{\psi_1}, L_{\psi_2}$ and with respect to (perhaps) different (semi)-norms $\|\cdot\|_{[a_1,b_1]},\|\cdot\|_{[a_2,b_2]}$.  Then, (with $``|"$ demarcating different definitions) $\psi_3 = \neg \psi_1 ~|~ \psi_1 \wedge \psi_2~|~\psi_1 \lor \psi_2$ all have Lipschitz robustness measures $\rho_{\psi_3}$.
\end{lemma}
\begin{proof}
    For $\psi_3 = \neg \psi_1$, $\rho_{\psi_3} = \rho_{\psi_1}$ suffices, as $\rho_{\psi_1}$ is assumed to be Lipschitz with constant $L_{\psi_1} \geq 0$ and (semi)-norm $\|\cdot\|_{[a_1,b_1]}$.  The same holds for $\psi_3 = \neg \psi_2$.  For $\psi_3 = \psi_1 \wedge \psi_2$, $\rho_{\psi_3}(s) = \min\{\rho_{\psi_1}(s),\rho_{\psi_2}(s)\}$ is Lipschitz with constant $L_{\psi_3} = \max\{L_{\psi_1},L_{\psi_2}\}$ and (semi)-norm $\|\cdot\|_{\min\{a_1,a_2\},\max\{b_1,b_2\}}$.  Likewise, for $\psi_3 = \psi_1 \lor \psi_2$, $\rho_{\psi_3}(s) = \max\{\rho_{\psi_1}(s),\rho_{\psi_2}(s)\}$ is Lipschitz with the same constant $L_{\psi_3}$ and (semi)-norm as prior. Proving these relations amounts to multiple applications of the Cauchy-Schwarz inequality while noting that $\min_{x \in X}f(x) = -\max_{x \in X}(-f(x))$ and that for arbitrary spaces $X$ and functions $f,g:X \to \mathbb{R}$, the following inequality is true:
    \begin{equation}
        \label{eq:absolute_max_diff}
        \left|\max_{x \in X}~f(x) - \max_{x \in X}~g(x)\right| \leq \max_{x \in X}~\left|f(x) - g(x)\right|.
    \end{equation}
\end{proof}

\noindent The following lemma~\ref{lem:until_conjunction} proves that if specifications $\psi_1,\psi_2$ have Lipschitz robustness measures, then $\psi_3 = \psi_1 \until_{[a_3,b_3]} \psi_2$ also has a robustness measure that is Lipschitz.

\begin{lemma}
\label{lem:until_conjunction}
    Let $\psi_1,\psi_2$ be STL specifications as per~\eqref{eq:spec} such that their corresponding robustness measures $\rho_{\psi_1},\rho_{\psi_2}$ are Lipschitz with constants $L_{\psi_1}, L_{\psi_2}$ and with respect to (perhaps) different (semi)-norms $\|\cdot\|_{[a_1,b_1]},\|\cdot\|_{[a_2,b_2]}$.  Then, $\psi_3 = \psi_1\until_{[a,b]}\psi_2$ has a Lipschitz robustness measure.
\end{lemma}
\begin{proof}
For notational ease, we abbreviate $|\rho_{\psi_i}(s,t) - \rho_{\psi_i}(z,t)| = \Delta_i(t)$ for signals $s,z \in \signalspace$. Consider the candidate robustness measure for $\psi_3$,
    \begin{equation}
        \rho_{\psi_3}(s,0) = \max_{t' \in [a,b]}\left(\min\left(\rho_{\psi_2}(s,t'), \min_{t'' \in [0,t']}\rho_{\psi_1}(s,t'')\right) \right).
    \end{equation}
To prove that $\rho_{\psi_3}$ is Lipschitz, equation~\eqref{eq:absolute_max_diff} is applied thrice, along with the fact $\min_{x \in X}~f(x) = -\max_{x \in X}(-f(x))$, to construct the following inequality:
\begin{equation}
    \Delta_3(0) \leq \max_{t' \in [a,b]}\left(\max\left(\Delta_2(t'), \max_{t'' \in [0,t']}\Delta_1(t'')\right) \right).
\end{equation}
Given that the component robustness measures $\rho_{\psi_1},\rho_{\psi_2}$ are Lipschitz, note that the inner arguments $\Delta_2(t')$ and $\Delta_1(t'')$ are dominated by their corresponding Lipschitz bounds.  While the bound has only been assumed for signals $s,z$ evaluated at $t=0$, we note that the robustness evaluation of any signal $s' \in \signalspace$ at some arbitrary time $\Tilde t \in \mathbb{R}_{\geq 0}$ is such that $\rho(s',\Tilde t) = \rho(s'',0)$ where $s''(0) = s'(\Tilde t)$. Setting $L_{\psi_3} = \max\{L_{\psi_1},L_{\psi_2}\}$, the following inequality is true: 
\begin{equation}
    \Delta_3(0) \leq L_{\psi_3} \|s-z\|_{\big[\underbrace{\min\{a_1,a_2+a\}}_{a_3}, ~\underbrace{b+\max\{b_1,b_2\}}_{b_3}\big]},
\end{equation}
thus completing the proof.
\end{proof}

\noindent The following theorem makes an inductive argument utilizing the Lemmas~\ref{lem:predicate}---\ref{lem:until_conjunction} to show that all STL specifications defined in~\eqref{eq:spec} have Lipschitz robustness measures.
\begin{theorem}\label{thm:stl_lipz}
All STL specifications $\psi$ as per the syntax~\eqref{eq:spec} have Lipschitz robustness measures $\rho_{\psi}$ as per Definition~\ref{def:robustness} with respect to a (semi)norm $\|\cdot\|_{[a,b]}$ as per Definition~\ref{def:signal_norm}.
\end{theorem}
\begin{proof}
    This proof follows an inductive argument. Consider the base case with specifications of the form $\psi = \mu_1 \until_{[a,b]} \mu_2$ with $\mu_1,\mu_2$ being propositions. Observe that \(\psi\) has Lipschitz robustness measures with respect to (semi)-norms over $\signalspace$.
    This follows from Lemmas~\ref{lem:predicate} and \ref{lem:until_conjunction}.  Note that Lemma~\ref{lem:predicate} proves that both propositions $\mu_1,\mu_2$ have robustness measures that are Lipschitz continuous with respect to the $2$-norm $\|\cdot\|$.  Furthermore, note that $\|\cdot\|_{[0,0]} = \|\cdot\|$.  Therefore, applying Lemma~\ref{lem:until_conjunction} in the case where the two specifications have Lipschitz continuous robustness measures with respect to the signal (semi)norm $\|\cdot\|_{[0,0]}$ proves the base case. Lemmas~\ref{lem:logical_conjunctions} and~\ref{lem:until_conjunction} provide the inductive steps. As the base and inductive steps are true, this completes the proof.
\end{proof}

\section{Control Barrier Function Synthesis}
\label{sec:cbf_synt}
We utilize these Lipschitz-continuous robustness measures for STL specifications to construct time-varying CBFs guaranteeing specification satisfaction to address Problem~\ref{prob}.  At a high level, this approach requires an expert signal $s$ satisfying the desired system specification $\psi$, though how that signal is generated is left to the practitioner.

\begin{remark}
We note that oftentimes it is easier to generate safe and satisfactory trajectories via reduced order models~\cite{molnar2021model}.  Therefore, the assumption of the existence of such a satisfying expert signal is (perhaps) easily satisfied.
\end{remark}

\noindent Formally, let the desired system specification be \(\psi\), and the expert trajectory be denoted by \(s \in \mathcal{S}\). By Theorem 2, we know that the robustness measure $\rho_{\psi}$ for the STL formula $\psi$ is Lipschitz-continuous with respect to some signal norm $\|\cdot\|_{[0,T]}$ and constant $L_{\psi} \geq 0$.  Note that as we only evaluate the robustness of state trajectories starting at time $t=0$, we denoted the signal (semi)norm to consider times $t \in [0, T]$, where the maximum time $T$ is determined by the specification of interest.  An example will be provided in a section to follow.  As a result, for all signals \(z \in \mathcal{S}^{\mathbb{R}^n}\), 
\begin{equation}\label{eq:lipz}
    |\rho_{\psi}(s, 0) - \rho_{\psi}(z,0)| \leq L_{\psi} \|s - z\|_{[0,T]}\, .
\end{equation}
\noindent Now consider signals $z$ satisfying the  following condition:
\begin{equation}
     L_\psi\|s-z\|_{[0,T]} \leq \rho_{\psi}(s,0),
      \label{eq:satisfying_traj}
\end{equation}
The set of all such signals $z$ constitutes a signal tube (around the expert trajectory \(s\)) of signals satisfying specification \(\psi\). If the system trajectory is enforced to remain within this signal tube, then the system will satisfy the specification \(\psi\). To enforce this condition, we propose the following time-varying control barrier function, 
    \begin{equation}
        h^s_{\psi}(x,t) = \rho_{\psi}(s,0)^2 - L_{\psi}^2\|x - s(t)\|^2.
        \label{eq:barrier_function}
    \end{equation}
Our main result proves that the proposed function \(h^s_{\psi}(x,t)\) is a valid time-varying control barrier function according to Definition~\ref{def:time_varying_cbf}. Therefore, by filtering control inputs to satisfy the inequality~\eqref{eq:tv_cbf_criteria}, we ensure positivity of the function $h^s_{\psi}$ and satisfaction of specification $\psi$.  To prove this result, we use an assumption from the works we extend~\cite{lindemann2018control,lindemann2019robust,lindemann2020barrier}:
\begin{assumption}\label{assump:pd}
The control forcing function \(g(x)\) is such that \(g(x).g(x)^{\top}\) is positive definite for all \(x \in \mathcal{X}\).
\end{assumption}
\begin{theorem}
\label{thm:valid_tvcbf}
Let Assumption~\ref{assump:pd} hold.  The function $h^s_{\psi}(x,t)$ as per~\eqref{eq:barrier_function} is a time-varying control barrier function as per Definition~\ref{def:time_varying_cbf} for our nonlinear control-affine system~\eqref{eq:base_system}.  Additionally, for some signal $z \in \signalspace$, 
\begin{equation}
h^s_{\psi}(z(t),t) \geq 0,~\forall~t \in [0,T] \implies z \models \psi.
\end{equation}
\end{theorem}
\begin{proof}
The first part of this proof will show that the candidate \(h^s_{\psi}(x,t)\) as defined in equation~\eqref{eq:barrier_function} is a time-varying control barrier function. 
To show this, note that the Lie derivative \( L_gh^s_{\psi} = \big(\frac{\partial h^s_{\psi}(x,t)}{\partial x}\big)^{\top}g(x) = \mathbf{0}^{\top}_m\) iff \(\frac{\partial h^s_{\psi}(x,t)}{\partial x} = \mathbf{0}_n\). This is due to Assumption~\ref{assump:pd} and the rank-nullity theorem which implies that the nullspace of \(g(x)\) is empty. As a result, we have two cases $\forall~(x,t)$, either (a) $L_gh^s_{\psi}(x,t) = \mathbf{0}_m$ or (b) $L_gh^s_{\psi} \neq \mathbf{0}_m$.  In (a), the prior line of reasoning indicates that $x = s(t)$, which results in $\dot h^s_{\psi}(x,t) = 0$ and $h^s_{\psi}(x,t) = \rho(s,0) \geq 0$.  As a result, the differential inequality~\eqref{eq:tv_cbf_criteria} collapses to $0 \geq 0$ for any choice of class-$\mathcal{K}_e$ function $\alpha$.  Therefore, any choice of input $u \in \mathbb{R}^m$ satisfies the differential inequality, so the supremum over all $u \in \mathbb{R}^m$ necessarily satisfies the same inequality.   In case (b), \(x\neq s_t\) and the following control input \(u\) ensures that the inequality~\eqref{eq:tv_cbf_criteria} holds for any class-$\mathcal{K}_e$ function $\alpha$:
\begin{equation}
    u = -(L_gh^s_{\psi})^{\dagger}\left(\alpha(h^s_{\psi}(x,t)) +\frac{\partial h^s_{\psi}}{\partial t} + L_fh^s_{\psi}\right),
\end{equation}
where \((L_gh^s_{\psi})^{\dagger}\) is the pseudo-inverse of \(L_gh^s_{\psi}\). Therefore, the supremum necessarily satisfies the same inequality.  As the result holds in both cases for any choice of class-$\mathcal{K}_e$ function $\alpha$, \(h^s_{\psi}(x,t)\) is a time-varying control barrier function.

The second part of this proof stems via Theorem~\ref{thm:stl_lipz}. If for some signal $z$, $h^s_{\psi}(z(t),t) \geq 0,~\forall~t \in [0,T]$, then by~\eqref{eq:lipz},
\begin{equation}
    L_{\psi}\|z - s\|_{[0,T]} \leq \rho_{\psi}(s,0) \implies \rho_{\psi}(z,0) \geq 0.
\end{equation}
By Definition~\ref{def:robustness} then, $z \models \psi$.
\end{proof}

\section{Reactive Path Planning Example}
\label{sec:examples}
This section illustrates our method applied to generate a reactive path planner for a single-agent system~\cite{wilson2020robotarium}. Here, we represent the true system via a unicycle model, \textit{i.e.} with system state $x \in \mathcal{X} \subseteq \mathbb{R}^3$ and control input $u \in \mathcal{U} \subset \mathbb{R}^2$:
\begin{equation}
    \label{eq:unicyle_model}
    \begin{aligned}
        \dot{x} & = \begin{bmatrix}
        \cos\left(x[3]\right) & 0 \\
        \sin\left(x[3]\right) & 0 \\
        0 & 1
        \end{bmatrix}u,~
        x = \begin{bmatrix}
        p_x \\
        p_y \\
        \theta
        \end{bmatrix},
        u = \begin{bmatrix}
        v \\
        \omega
        \end{bmatrix}.
    \end{aligned}
\end{equation}
Furthermore, we aim to generate a controller for this system that always identifies feasible paths to a goal, in a randomly generated maze-like setting with moving obstacles. More accurately, over a subset of the plane $\mathcal{W} = [-1.6,1.6] \times [-1,1]$, we define the set of considered operating environments $E$ as follows: (1) there must exist $8$ static obstacles occupying different cells in an $8 \times 5$ evenly-spaced grid overlaid on $\mathcal{W}$; (2) the starting locations of $4$ moving obstacles $o_i \in \mathcal{W}$ and the ego agent's initial state $x_0 \in \mathcal{X}$ must occupy different cells and not exist within static obstacles; (3) there must exist $3$ goal cells for the agent with at least one feasible path existing to a goal at initialization --- the center of each goal cell $g_i \in \mathcal{W}$; (4) the moving agents stop when they reach within $0.2$ meters of the ego agent, and can only progress in their motion if the agent moves away --- this is to prevent faster-moving obstacles from crashing into the ego-agent without the agent having the ability to react.  Figure~\ref{fig:expert_controller} depicts an example setup of a valid environment $E$.

To formally define our specification $\psi$, we first define $SO \subset \mathcal{W}$ to be the set of planar centers of the $8$ static obstacles in our environment $E$.  Second, we define $P_{\odot}:\mathcal{W} \to \mathbb{R}_{\geq 0}$ to be the path distance function defined as follows: (1) Identify the shortest feasible path --- with respect to the number of cells that must be traversed --- between the planar projection of the initial condition $x_0$ and a goal in the set $\odot \subset \mathcal{W}$; (2) record the center $c$ of the corresponding goal and output $\|w-c\|$.  The formal system specification $\psi$ is as follows, with $\Pi = [\mathbf{I}_2,\mathbf{0}_{2 \times 1}]$:
\begin{equation}
    \label{eq:example_spec}
    \psi = (B \lor \psi_G) \wedge (\neg B \lor \psi_H).
\end{equation}
The subformulas $B$ and $\psi_{\odot}$ are defined as follows:
\begin{itemize}
    \item[($\psi_{\odot}$)] \emph{Eventually} make non-trivial progress towards the closest feasible goal in the set $\odot$ over a $10$ second period.  Here $G = \{g_1,g_2,g_3\}$ corresponds to the randomly generated goals as part of the environment $E$ as previously described, and $H = \{h_1,h_2,h_3,h_4\}$ corresponds to the centers of four home cells at the corners of the $8 \times 5$ grid. Also, \emph{always} avoid static and moving obstacles and treat moving obstacles to be frozen at the time instant of evaluation.  Mathematically, a waypoint signal $\mathbf{w}$ taking values in $\mathcal{W}$ satisfies $\psi_{\odot}$ at time $t$, iff
        \begin{enumerate}
            \item $P_{\odot}(\Pi x_0) - P_{\odot}(\mathbf{w}(t+10)) \geq 0$, and
            \item $\forall~t' \in [t,t+10]$, $\min_{o \in SO}~\|\mathbf{w}(t') - o\|_{\infty} \geq 0.2$m and $\min_{i=1,2,3,4} \|\mathbf{w}(t') - o_i(t)\| \geq 0.18$m.
        \end{enumerate}
    \item[($B$)] $B$ is an atomic proposition related to the existence of a basing signal.  If the signal exists at time $t$ then $B(t) = \true$.  If not, then $B(t) = \false$.
\end{itemize}
In short, \(\psi\) represents the requirement that the system switch between two different sets of goals, \(G\) and \(H\), upon receiving a basing signal \(B\), while also avoiding obstacles.
\begin{figure}[t]
    \centering
    \vspace{0.075 in}
    \includegraphics[width = \columnwidth]{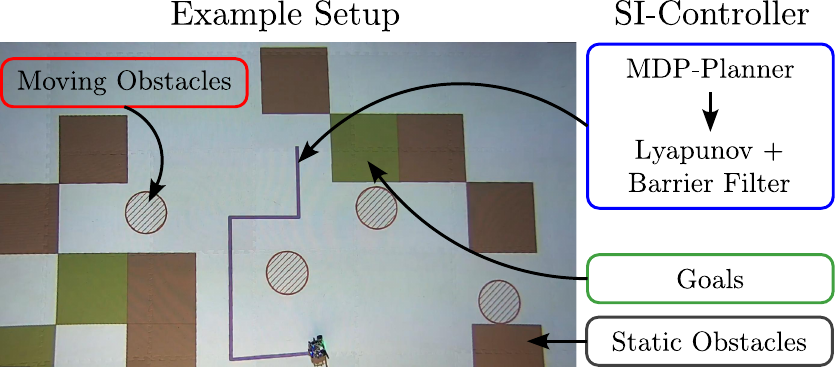} \vspace{-0.2 in}
    \caption{Expert controller architecture for reduced-order single-integrator system utilized to generate signals in Section~\ref{sec:examples}, and an example hardware setup of an environment considered in the same section.} \vspace{-0.2 in}
    \label{fig:expert_controller}
\end{figure}

\newidea{Expert Signal Construction} To construct an expert signal $s$ satisfying our desired specification $\psi$~\eqref{eq:example_spec}, we queried a controlled single-integrator system.  Figure~\ref{fig:expert_controller} depicts the expert system's control architecture along with an example setting consistent with the prior description.  Through exhaustive Monte Carlo testing, we determined that this controller architecture should always produce a single integrator waypoint trajectory $\mathbf{w}\models \psi$. To construct our expert trajectory $s$ that takes values in $\mathbb{R}^3$, we define $s(t)$ as follows:
\begin{equation}
    s(t) = \begin{bmatrix}
        \mathbf{w}(t) \\
        \arctan\left(\frac{\dot{\mathbf{w}}(t)[1]}{\dot{\mathbf{w}}(t)[0]}\right)
    \end{bmatrix},~\mathbf{w}(t) = \begin{bmatrix}
        p^{\mathrm{si}}_x(t) \\
        p^{\mathrm{si}}_y(t)
    \end{bmatrix}.
\end{equation}

\newidea{CBF Construction}  To define our time-varying CBF filter, we need to evaluate the robustness of the expert trajectory $s$ determined prior.  To construct this robustness measure $\rho_{\psi}$, we define an obstacle avoidance function $OA$:
\begin{equation}
    \label{eq:obstacle_avoidance}
    OA(x,t) = \min \left(
    \begin{gathered}
        \min_{o \in SO} \|x-o\|_{\infty} - 0.2, \\
        \min_{i=1,2,3,4} \|x - o_i(t)\| -0.18.
    \end{gathered}
    \right)
\end{equation}
and a path difference function $\Delta P_{\odot}$ where:
\begin{equation}
    \label{eq:lyapunov_difference}
    \Delta P_{\odot}(s,t) = P_{\odot}(\Pi x_0) - P_{\odot}(\Pi s(t+10)).
\end{equation}
Then, our robustness measure $\rho_{\psi}$ is as follows:
\begin{gather}
    \Phi(s,t,\odot) = \min\left(\Delta P_{\odot}(s,t), \min_{t' \in [t,t+10]} OA(\Pi s(t'),t)  \right)
    \\
    \label{eq:example_robustness}
    \rho_{\psi}(s,t) = \begin{cases}
    \Phi(s,t,H) & \mbox{if~}B(t) = \true, \\
    \Phi(s,t,G) & \mbox{if~}B(t) = \false.
    \end{cases}
\end{gather}

\begin{figure*}[t]
    \centering
    \vspace{0.075in}
    \includegraphics[width = \textwidth]{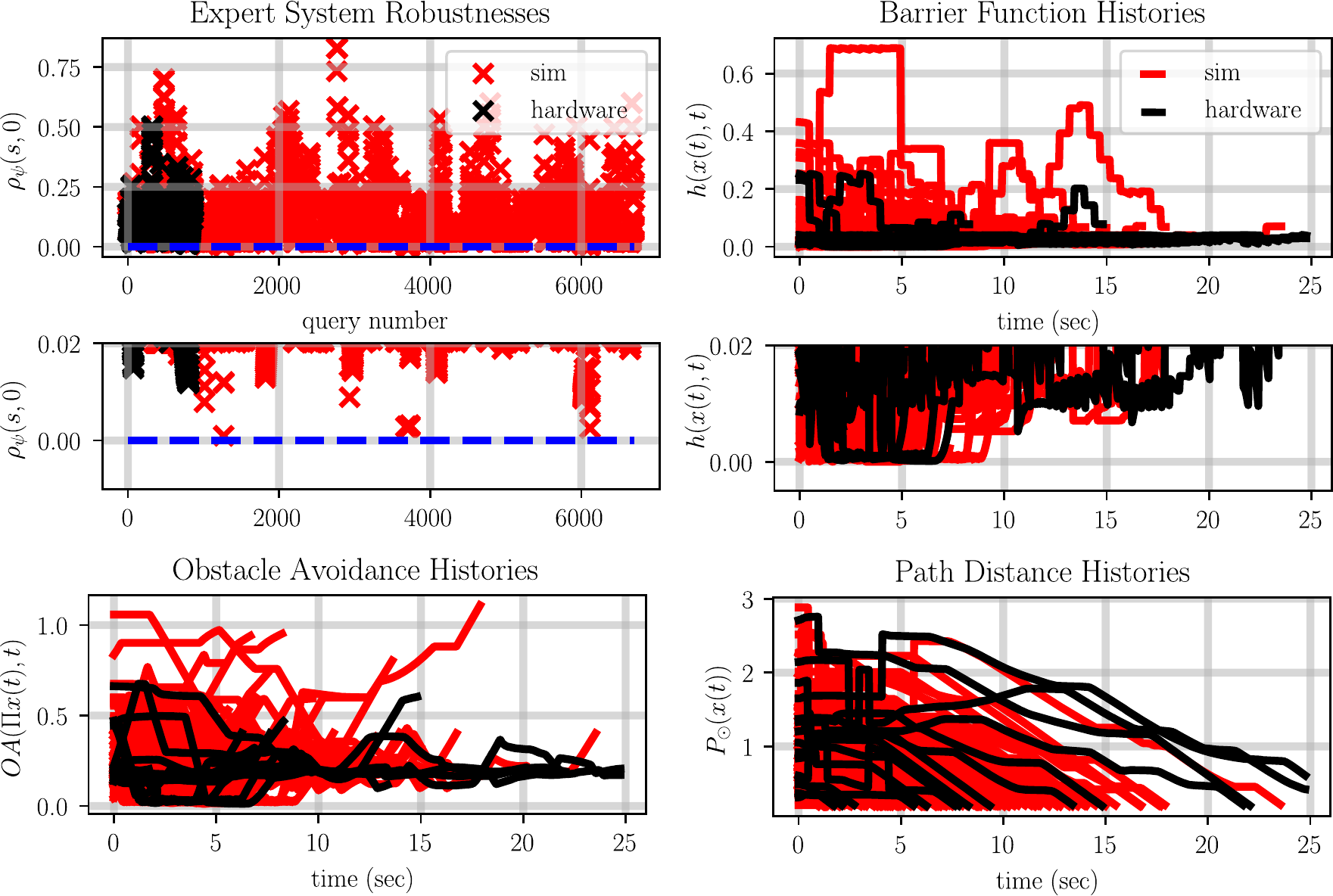} \vspace{-0.2 in}
    \caption{Data corresponding to $145$ simulation trials (red) and $15$ hardware trials (black) of a unicycle system with a controller filtered by the synthesized time-varying CBF. Notice that every time we queried the expert system, the expert signal's robustness is $\rho_{\psi}(s,0) \geq 0$ as emphasized in the two, top left stacked figures, with the bottom a magnification of robustness values between $0$ and $0.02$.  For each valid expert signal, we constructed a barrier function $h^s_{\psi}$ as per~\eqref{eq:example_barrier_filter}.  The time trajectories of these barrier functions are depicted on the top right, with a magnification of values between $0$ and $0.02$ on the middle right.  Notice that the barrier value was always positive.  Per our theory, this should imply that our true system always satisfied the specification $\psi$ --- avoid obstacles and make non-trivial progress towards a goal within $10$ seconds.  (Bottom Left) The histories of the obstacle avoidance function $OA$~\eqref{eq:obstacle_avoidance} for our true system for each trial.  This function was always positive indicating the system never crashed into any static or moving obstacles, as theorized.  (Bottom Right) The histories of the path distance function $P_{\odot}$ described in our specification description~\eqref{eq:example_spec}.  Notice that even if the goals change and there is a momentary sharp increase in the path distance to a goal, the system always makes non-trivial progress towards its goals as all distances decay to $0$ eventually.} \vspace{-0.2 in}
    \label{fig:sim_data}
\end{figure*}

Notice that $\rho_{\psi}(s,0) \geq 0 \iff (s,0) \models \psi$ insofar as we're checking for positivity of the same conditions as expressed in the mathematical formalization of our specification $\psi$~\eqref{eq:example_spec}.  Furthermore, we also know that $\rho_{\psi}$~\eqref{eq:example_robustness} is Lipschitz continuous when evaluated at time $t=0$ with constant $L_{\psi}=1$ and signal (semi)norm $\|\cdot\|^Q_{[0,10]}$.  Here, $\|x\|^Q = \sqrt{x^TQx}$ where $Q = \mathrm{diag}(\mathbf{I}_{2 \times 2}, 0)$ and $\|s\|^Q_{[0,10]} = \max_{t \in [0,10]}\|s(t)\|^Q$.  To show this, we know the obstacle avoidance function~\eqref{eq:obstacle_avoidance} is Lipschitz continuous with respect to the Euclidean $2$-norm over $\mathbb{R}^2$ and constant $L=1$.  Here, we note that for any vector $v \in \mathbb{R}^n$, $\|v\|_{\infty} \leq \|v\|$.  Second, the path difference function is also Lispchitz continuous with respect to the Euclidean $2$-norm over $\mathbb{R}^2$ and constant $L=1$.  Then, the signal (semi)norm arises via the minimization over time in~\eqref{eq:example_robustness}, noting that for any signal $z$ taking values in $\mathbb{R}^n$, $\|z(t+10)\| \leq \|z\|_{[0,10]}$ by Definition~\ref{def:signal_norm}.  Finally, we note that for any vector $v \in \mathbb{R}^3$, $\|v\|^Q \leq \|v\|$ and therefore for any signal $z$ taking values in $\mathbb{R}^3$, $\|z\|^Q_{[0,10]} \leq \|z\|_{[0,10]}$.  As such, our robustness measure $\rho_{\psi}$ still satisfies the conditions of Theorem~\ref{thm:stl_lipz}, though we note that using the signal (semi)norm $\|\cdot\|^Q_{[0,10]}$ as opposed to $\|\cdot\|_{[0,10]}$ increases performance as will be detailed below.  As a result, our time-varying CBF according to Theorem~\ref{thm:valid_tvcbf} is as follows:
\begin{equation}
    \label{eq:example_barrier_filter}
    h^s_{\psi}(x,t) = \rho_{\psi}(s,0)^2 - (x-s(t))^T Q (x-s(t)).
\end{equation}

\newidea{Implementation:} We initialize an environment $E$ consistent with the prior description, query the expert, controlled single-integrator system to produce an initial expert signal $s$, and define our unicycle agent's controller to be a Lyapunov-based controller filtered via the barrier function in~\eqref{eq:example_barrier_filter}.  Mathematically, we define $x(t)$ to be the state of the unicycle system at time $t$, $V(t) = \|x(t) -s(t)\|$, and $u_v(t)$ to be the control Lyapunov input that ensures that $\dot V(t) = -2V(t)$.  Then, we define the filtered, un-bounded input $u(t)$ as the solution to the following optimization problem:
\begin{align}
    u(t) & = \argmin_{u \in \mathbb{R}^m}~& &\|u - u_v(t)\| ,  \label{eq:example_controller} \\
    & ~~\mathrm{subject~to}~ & & \dot h^s_{\psi}(x(t),t,u) \geq -2h^s_{\psi}(x(t),t).
\end{align}
Finally, the controller $U$ is constructed by projecting this input $u(t)$ to the input space $\mathcal{U} = [-0.2,0.2] \times [-\frac{\pi}{4},\frac{\pi}{4}]$ at every time step \(t\).
We zero-order hold that input for a $\Delta t = 0.033$ second interval, simulating control inputs provided to our system at $30$ Hz --- the update frequency of the Robotarium.  Finally, we re-query the expert system every $0.25$ seconds to account for moving obstacles, resetting the expert signal $s$.  At every reset time $t_r$, we reset the signal time to $0$ but keep track of true system time separately and account for the shifted time in our Lyapunov controller. Finally, the initial condition $x_0$ for the path difference function $P_{\odot}$ utilized to define the robustness measure~\eqref{eq:example_robustness} is reset to the unicycle state at the reset time, \textit{i.e.} $x_0 = x(t_r)$ for all $t_r$.

\newidea{Figure Analysis:} In Figure~\ref{fig:sim_data}, the data in red corresponds to $145$ separate simulation trials of controlling our unicycle system with the CBF-filtered controller~\eqref{eq:example_controller}, and the data in black corresponds to $15$ hardware demonstrations on the Robotarium. If the expert signal $s$ realizes positive robustness $\rho_{\psi}(s,0)$, then we can construct a time-varying CBF $h^s_{\psi}(x,t)$ given in equation~\eqref{eq:example_barrier_filter} to filter control inputs and ensure the system satisfies specification $\psi$~\eqref{eq:example_spec}.  Notice that for all trials, the expert signal robustness $\rho_{\psi}(s,0)$ is non-negative, and the corresponding CBF filter $h^s_{\psi}(x(t),t) \geq 0,~\forall~t \geq 0$ until the run is terminated.  As a result, we expect the system to always satisfy our specification $\psi$~\eqref{eq:example_spec} --- avoid obstacles and make non-trivial progress to a goal within $10$ seconds.  The corresponding data is portrayed at the bottom in Figure~\ref{fig:sim_data}.  Notice that the trajectory of the obstacle avoidance function $OA$~\eqref{eq:obstacle_avoidance}, when evaluated over time for our true system trajectory $x(t)$, is always positive, both in simulation and on hardware.  As a result, for every simulation and hardware trial, the system did not crash into any static or moving obstacles.  Furthermore, the path distance function trajectories indicate that even if the goals change --- leading to a momentary sharp jump in the path distance --- the system always takes steps to move towards the goal over a $10$ second period.  This eventually results in the system always reaching its goal as all paths converge to values between $0$ and $0.2$ (half the length of the side of a goal cell).  To note, two hardware trajectories timed out before the system could reach the goal cell.  However, the system still made non-trial progress toward its goal overall and would have reached the goal provided more time.  In summary, the data implies that not only can we synthesize barrier functions from one expert demonstration by leveraging the Lipschitz continuity of signal temporal logic robustness measures, but we can also steer systems using these barrier functions to ensure the satisfaction of the same signal temporal logic specification, both in simulation and on hardware.

\section{Conclusion and Future Work}
Through an inductive argument, we proved that every signal temporal logic specification has a Lipschitz continuous robustness measure with respect to a (semi)norm over the space of all signals.  Leveraging this information, we provided a constructive method to determine control barrier functions for arbitrary systems and showcased the ability of these synthesized functions to realize the specification. For future work, we would first like to extend our existence results to outline a constructive method to determine these Lipschitz-continuous robustness measures. This would facilitate control barrier function synthesis.  Second, as shown in the examples, repetitive signal generation and controller modification seem to generate reliable behavior.  As such, there might be a potential interface with Model Predictive Control that we would like to explore. Finally, we would like to account for uncertainties to robustify hardware implementation.
\balance
\bibliographystyle{IEEEtran}
\bibliography{IEEEabrv,bib_works}

\begin{thebibliography}{10}
\providecommand{\url}[1]{#1}
\csname url@rmstyle\endcsname
\providecommand{\newblock}{\relax}
\providecommand{\bibinfo}[2]{#2}
\providecommand\BIBentrySTDinterwordspacing{\spaceskip=0pt\relax}
\providecommand\BIBentryALTinterwordstretchfactor{4}
\providecommand\BIBentryALTinterwordspacing{\spaceskip=\fontdimen2\font plus
\BIBentryALTinterwordstretchfactor\fontdimen3\font minus
  \fontdimen4\font\relax}
\providecommand\BIBforeignlanguage[2]{{%
\expandafter\ifx\csname l@#1\endcsname\relax
\typeout{** WARNING: IEEEtran.bst: No hyphenation pattern has been}%
\typeout{** loaded for the language `#1'. Using the pattern for}%
\typeout{** the default language instead.}%
\else
\language=\csname l@#1\endcsname
\fi
#2}}

\bibitem{ames2016control}
A.~D. Ames, X.~Xu, J.~W. Grizzle, and P.~Tabuada, ``Control barrier function
  based quadratic programs for safety critical systems,'' \emph{IEEE
  Transactions on Automatic Control}, vol.~62, no.~8, pp. 3861--3876, 2016.

\bibitem{hsu2015control}
S.-C. Hsu, X.~Xu, and A.~D. Ames, ``Control barrier function based quadratic
  programs with application to bipedal robotic walking,'' in \emph{2015
  American Control Conference (ACC)}.\hskip 1em plus 0.5em minus 0.4em\relax
  IEEE, 2015, pp. 4542--4548.

\bibitem{xiong2021slip}
X.~Xiong and A.~Ames, ``Slip walking over rough terrain via h-lip stepping and
  backstepping-barrier function inspired quadratic program,'' \emph{IEEE
  Robotics and Automation Letters}, vol.~6, no.~2, pp. 2122--2129, 2021.

\bibitem{srinivasan2018control}
M.~Srinivasan, S.~Coogan, and M.~Egerstedt, ``Control of multi-agent systems
  with finite time control barrier certificates and temporal logic,'' in
  \emph{2018 IEEE Conference on Decision and Control (CDC)}.\hskip 1em plus
  0.5em minus 0.4em\relax IEEE, 2018, pp. 1991--1996.

\bibitem{zeng2021safety}
J.~Zeng, B.~Zhang, and K.~Sreenath, ``Safety-critical model predictive control
  with discrete-time control barrier function,'' in \emph{2021 American Control
  Conference (ACC)}.\hskip 1em plus 0.5em minus 0.4em\relax IEEE, 2021, pp.
  3882--3889.

\bibitem{choi2020reinforcement}
J.~Choi, F.~Castaneda, C.~J. Tomlin, and K.~Sreenath, ``Reinforcement learning
  for safety-critical control under model uncertainty, using control lyapunov
  functions and control barrier functions,'' \emph{arXiv preprint
  arXiv:2004.07584}, 2020.

\bibitem{kolathaya2018input}
S.~Kolathaya and A.~D. Ames, ``Input-to-state safety with control barrier
  functions,'' \emph{IEEE control systems letters}, vol.~3, no.~1, pp.
  108--113, 2018.

\bibitem{alan2021safe}
A.~Alan, A.~J. Taylor, C.~R. He, G.~Orosz, and A.~D. Ames, ``Safe controller
  synthesis with tunable input-to-state safe control barrier functions,''
  \emph{IEEE Control Systems Letters}, vol.~6, pp. 908--913, 2021.

\bibitem{tezuka2020time}
I.~Tezuka and H.~Nakamura, ``Time-varying obstacle avoidance by using high-gain
  observer and input-to-state constraint safe control barrier function,''
  \emph{IFAC-PapersOnLine}, vol.~53, no.~5, pp. 391--396, 2020.

\bibitem{cai2021safe}
Z.~Cai, H.~Cao, W.~Lu, L.~Zhang, and H.~Xiong, ``Safe multi-agent reinforcement
  learning through decentralized multiple control barrier functions,''
  \emph{arXiv preprint arXiv:2103.12553}, 2021.

\bibitem{tan2021distributed}
X.~Tan and D.~V. Dimarogonas, ``Distributed implementation of control barrier
  functions for multi-agent systems,'' \emph{IEEE Control Systems Letters},
  vol.~6, pp. 1879--1884, 2021.

\bibitem{marvi2021safe}
Z.~Marvi and B.~Kiumarsi, ``Safe reinforcement learning: A control barrier
  function optimization approach,'' \emph{International Journal of Robust and
  Nonlinear Control}, vol.~31, no.~6, pp. 1923--1940, 2021.

\bibitem{ma2021model}
H.~Ma, J.~Chen, S.~Eben, Z.~Lin, Y.~Guan, Y.~Ren, and S.~Zheng, ``Model-based
  constrained reinforcement learning using generalized control barrier
  function,'' in \emph{2021 IEEE/RSJ International Conference on Intelligent
  Robots and Systems (IROS)}.\hskip 1em plus 0.5em minus 0.4em\relax IEEE,
  2021, pp. 4552--4559.

\bibitem{jagtap2020control}
P.~Jagtap, G.~J. Pappas, and M.~Zamani, ``Control barrier functions for unknown
  nonlinear systems using gaussian processes,'' in \emph{2020 59th IEEE
  Conference on Decision and Control (CDC)}.\hskip 1em plus 0.5em minus
  0.4em\relax IEEE, 2020, pp. 3699--3704.

\bibitem{robey2020learning}
A.~Robey, H.~Hu, L.~Lindemann, H.~Zhang, D.~V. Dimarogonas, S.~Tu, and
  N.~Matni, ``Learning control barrier functions from expert demonstrations,''
  in \emph{2020 59th IEEE Conference on Decision and Control (CDC)}.\hskip 1em
  plus 0.5em minus 0.4em\relax IEEE, 2020, pp. 3717--3724.

\bibitem{srinivasan2020synthesis}
M.~Srinivasan, A.~Dabholkar, S.~Coogan, and P.~A. Vela, ``Synthesis of control
  barrier functions using a supervised machine learning approach,'' in
  \emph{2020 IEEE/RSJ International Conference on Intelligent Robots and
  Systems (IROS)}.\hskip 1em plus 0.5em minus 0.4em\relax IEEE, 2020, pp.
  7139--7145.

\bibitem{lindemann2018control}
L.~Lindemann and D.~V. Dimarogonas, ``Control barrier functions for signal
  temporal logic tasks,'' \emph{IEEE control systems letters}, vol.~3, no.~1,
  pp. 96--101, 2018.

\bibitem{lindemann2020barrier}
------, ``Barrier function based collaborative control of multiple robots under
  signal temporal logic tasks,'' \emph{IEEE Transactions on Control of Network
  Systems}, vol.~7, no.~4, pp. 1916--1928, 2020.

\bibitem{huang2020multi}
X.~Huang, L.~Li, and J.~Chen, ``Multi-agent system motion planning under
  temporal logic specifications and control barrier function,'' \emph{Control
  Theory and Technology}, vol.~18, no.~3, pp. 269--278, 2020.

\bibitem{nguyen2016exponential}
Q.~Nguyen and K.~Sreenath, ``Exponential control barrier functions for
  enforcing high relative-degree safety-critical constraints,'' in \emph{2016
  American Control Conference (ACC)}.\hskip 1em plus 0.5em minus 0.4em\relax
  IEEE, 2016, pp. 322--328.

\bibitem{azimi2021exponential}
V.~Azimi and S.~Hutchinson, ``Exponential control lyapunov-barrier function
  using a filtering-based concurrent learning adaptive approach,'' \emph{IEEE
  Transactions on Automatic Control}, 2021.

\bibitem{dawson2022safe}
C.~Dawson, S.~Gao, and C.~Fan, ``Safe control with learned certificates: A
  survey of neural lyapunov, barrier, and contraction methods,'' \emph{arXiv
  preprint arXiv:2202.11762}, 2022.

\bibitem{donze2010robust}
A.~Donz{\'e} and O.~Maler, ``Robust satisfaction of temporal logic over
  real-valued signals,'' in \emph{International Conference on Formal Modeling
  and Analysis of Timed Systems}.\hskip 1em plus 0.5em minus 0.4em\relax
  Springer, 2010, pp. 92--106.

\bibitem{wilson2020robotarium}
S.~Wilson, P.~Glotfelter, L.~Wang, S.~Mayya, G.~Notomista, M.~Mote, and
  M.~Egerstedt, ``The robotarium: Globally impactful opportunities, challenges,
  and lessons learned in remote-access, distributed control of multirobot
  systems,'' \emph{IEEE Control Systems Magazine}, vol.~40, no.~1, pp. 26--44,
  2020.

\bibitem{akella2022sample}
P.~Akella, A.~Dixit, M.~Ahmadi, J.~W. Burdick, and A.~D. Ames, ``Sample-based
  bounds for coherent risk measures: Applications to policy synthesis and
  verification,'' \emph{arXiv preprint arXiv:2204.09833}, 2022.

\bibitem{akella2022scenario}
P.~Akella, M.~Ahmadi, and A.~D. Ames, ``A scenario approach to risk-aware
  safety-critical system verification,'' \emph{arXiv preprint
  arXiv:2203.02595}, 2022.

\bibitem{forsgren2002interior}
A.~Forsgren, P.~E. Gill, and M.~H. Wright, ``Interior methods for nonlinear
  optimization,'' \emph{SIAM review}, vol.~44, no.~4, pp. 525--597, 2002.

\bibitem{baier2008principles}
C.~Baier and J.-P. Katoen, \emph{Principles of model checking}.\hskip 1em plus
  0.5em minus 0.4em\relax MIT press, 2008.

\bibitem{fainekos2009robustness}
G.~E. Fainekos and G.~J. Pappas, ``Robustness of temporal logic specifications
  for continuous-time signals,'' \emph{Theoretical Computer Science}, vol. 410,
  no.~42, pp. 4262--4291, 2009.

\bibitem{maler2004monitoring}
O.~Maler and D.~Nickovic, ``Monitoring temporal properties of continuous
  signals,'' in \emph{Formal Techniques, Modelling and Analysis of Timed and
  Fault-Tolerant Systems}.\hskip 1em plus 0.5em minus 0.4em\relax Springer,
  2004, pp. 152--166.

\bibitem{molnar2021model}
T.~G. Molnar, R.~K. Cosner, A.~W. Singletary, W.~Ubellacker, and A.~D. Ames,
  ``Model-free safety-critical control for robotic systems,'' \emph{IEEE
  robotics and automation letters}, vol.~7, no.~2, pp. 944--951, 2021.

\bibitem{lindemann2019robust}
L.~Lindemann and D.~V. Dimarogonas, ``Robust control for signal temporal logic
  specifications using discrete average space robustness,'' \emph{Automatica},
  vol. 101, pp. 377--387, 2019.

\end{thebibliography}

\end{document}